\documentclass[onecolumn]{IEEEtran}
\usepackage{multicol}
\usepackage{cite}
\usepackage{amsmath}
\usepackage{amsfonts}
\usepackage{pifont}
\usepackage{amssymb}
\usepackage{amsthm}
\usepackage{tikz}
\usepackage{cases}
\usepackage{graphicx}
\usepackage{float,stfloats}
    \graphicspath{{../}}
    \DeclareGraphicsExtensions{.pdf}
\usepackage[justification=centering]{caption}
\usepackage[caption=true,font=footnotesize]{subfig}
\usepackage{multirow}
\usepackage{booktabs}
\usepackage{url}
\usepackage{xtab}
\usepackage{tabu}
\usepackage{longtable}
\usepackage{algorithm}
\usepackage{algorithmic}
\usepackage{enumerate}
\usepackage{makecell}
\usepackage{lipsum}
\usepackage{multicol}
\usepackage{mathdots}
\usepackage{extarrows}
\usepackage{CJKutf8}
\usepackage{color,xcolor}
\usepackage{bm} 
\usepackage{array}

\usetikzlibrary{arrows}
\theoremstyle{plain}

\newtheorem{theorem}{Theorem}

\newtheorem{proposition}[theorem]{Proposition}
\newtheorem{definition}[theorem]{Definition}
\newtheorem{corollary}[theorem]{Corollary}
\theoremstyle{definition}
\newtheorem{example}{Example}
\newtheorem{construction}{Construction}
\newtheorem{remark}{Remark}

\def\BibTeX{{\rm B\kern-.05em{\sc i\kern-.025em b}\kern-.08em
    T\kern-.1667em\lower.7ex\hbox{E}\kern-.125emX}}

\begin{document}
\title{Rack-Aware Regenerating Codes with Multiple Erasure Tolerance\thanks{This paper was presented in part at the 2021 International Symposium on Information Theory \cite{ZhouarXiv20}}}
%
%
%

\author{Liyang Zhou,
        Zhifang Zhang
\thanks{Liyang Zhou and Zhifang Zhang are both with KLMM, Academy of Mathematics and Systems Science, Chinese Academy of Sciences, and
School of Mathematical Sciences, University of Chinese Academy of Sciences, Beijing, China. e-mail: zhouliyang17@mails.ucas.ac.cn,~~zfz@amss.ac.cn}}

\maketitle

\begin{abstract}
We study the rack-aware storage system where all storage nodes are organized in racks and within each rack the nodes can communicate freely without taxing the system bandwidth.
Rack-aware regenerating codes (RRCs) were proposed for minimizing the repair bandwidth for single erasures. In the initial setting of RRCs, the repair of a single node requires the participation of all the remaining nodes in the rack containing the failed node as well as a large number of helper racks containing no failures. Consequently, the repair may be infeasible in front of multiple node failures.
In this work, a relaxed repair model
that can tolerate multiple node failures by simultaneously reducing the intra-rack connections and cross-rack connections is proposed. A tradeoff between the storage and repair bandwidth
under the relaxed repair model is derived, and parameters of the two extreme points on the tradeoff curve are
characterized for the minimum storage and minimum bandwidth respectively. Moreover, two codes corresponding to the extreme points are explicitly constructed over the fields of size comparable to the code length and with the lowest sub-packetization.
Finally, for the convenience of practical use, systematic encoding processes for the two codes are also established.
\end{abstract}


\section{Introduction}\label{sec0}
In distributed storage systems,  node failures occur frequently due to unexpected situations, such as server downtime, power outage, network disconnection, etc. In order to maintain data reliability, a self-sustaining system must repair node failures (i.e., recover the data stored in failed nodes) by downloading data from surviving nodes. A crucial metric for node repair efficiency is the repair bandwidth, i.e., the total amount of data transmitted during the repair process. Erasure codes attaining the optimal repair bandwidth have attracted a lot of study \cite{Balaji2018Survey} in the past decade.

The traditional work usually assumes a homogeneous network where the communication cost between all nodes are equal.
However, modern data centers often have hierarchical topologies that the nodes are organized in racks (or clusters) and the cross-rack (or inter-cluster) communication cost is much more expensive than the intra-rack (or intra-cluster) communication cost. Characterizing the optimal node repair in hierarchical data centers becomes a newly arisen branch in erasure codes which have been discussed in different rack-aware (or clustered) storage models \cite{Hu16,Prakash17,Prakash18,Sohn18,Sohn19,Hou19}. In this work, we focus on the rack-aware storage model adopted in \cite{Hu16,Hou19,Chen,Hou20,ZhouMSRR,ZhouarXivMBRR,Wangjj}.

Specifically, suppose $n=\bar{n}u$ and the $n$ nodes are organized in $\bar{n}$ racks each containing $u$ nodes. A data file consisting of $B$ symbols is stored across the $n$ nodes each storing $\alpha$ symbols such that any $k$ nodes can retrieve the data file.
When a node fails, the node repair process is to generate a replacement node that stores exactly the same data as the failed node does.
To rule out the trivial case, we assume throughout that $k\geq u$, so that any single node failure cannot be repaired locally within a rack. With respect to a repair process of a failed node, the rack containing the failed node is called  the {\it host rack}. The repair process is accomplished by the two kinds of communication below:
\begin{enumerate}
   \item \textbf{Intra-rack transmission.}
	The remaining $u-1$ nodes in the host rack transmit information to the replacement node.
	\item \textbf{Cross-rack transmission.}
	Outside the host rack, $\bar{d}$ helper racks each transmit $\beta$ symbols to the replacement node.
\end{enumerate}
Assume the nodes within each rack can communicate freely without taxing the system bandwidth, thus the repair bandwidth $\gamma$ only dependents on the cross-rack transmission, i.e., $\gamma=\bar{d}\beta$.

The above rack-aware storage model was introduced in \cite{Hu16} where the authors derived a lower bound on the repair bandwidth for codes with the minimum storage and also presented an existential construction of codes attaining this lower bound. Hou et al. \cite{Hou19} derived the cut-set bound for this model and further characterized a tradeoff curve between $\alpha$ and $\beta$. The codes with parameters lying on the tradeoff curve are called rack-aware regenerating codes (RRCs). Two extreme points on the tradeoff curve corresponding to the minimum storage rack-aware regenerating (MSRR) codes and minimum bandwidth rack-aware regenerating (MBRR) codes are of special interests. The authors of \cite{Hou19} presented existential constructions of MBRR and MSRR codes over sufficiently large fields. Explicit constructions of MSRR codes for all admissible parameters were derived in \cite{Chen}. Then, Hou et al. \cite{Hou20} and Zhou et al. \cite{ZhouMSRR} further reduced the sub-packetization (i.e., the number of symbols stored in each node) for MSRR codes. Recently, Zhou et al. \cite{ZhouarXivMBRR} developed explicit constructions of MBRR codes for all parameters.

Unfortunately, due to repair model assumed in RRCs, the repair is fragile in front of multiple node failures. Specifically, the repair of a failed node requires the participation of all the remaining $u-1$ nodes in the host rack as well as $\bar{d}$ helper racks that contain no failures at all.
As a result, when more than one node fails in a rack, the repair process designed for RRCs \cite{Hou19,Chen,Hou20,ZhouMSRR,ZhouarXivMBRR} can no longer work. Moreover, in all the existing RRCs \cite{Hou19,Chen,Hou20,ZhouMSRR,ZhouarXivMBRR}, it always assumes $\bar{d}\geq\bar{k}$ where $\bar{k}=\lfloor\frac{k}{u}\rfloor$. To reduce the storage overhead, $\bar{k}$ must be very near to $\bar{n}$. Thus the assumption $\bar{d}\geq\bar{k}$ usually leads to $\bar{d}=\bar{n}-1$ or $\bar{n}-2$. Therefore, in case of multiple node failures, it may fail to find $\bar{d}$ helper racks that contain no failures at all. In other words, the node repair process of existing RRCs mostly can tolerate only one node failure.

However, the scenarios of multiple node failures are quite common in distributed storage systems. For example, many systems (e.g., Total Recall \cite{Totalrecall}) apply a lazy repair strategy where the repair process is triggered only after the total number of failed nodes has reached a predefined threshold. Thus it is necessary to design a repair process that can tolerate multiple node failures. For this purpose, we relax the initial repair model of RRCs from the following two aspects.

\vspace{-10pt}
\subsection*{(a) Reducing the cross-rack connections by restricting to $\bar{d}<\bar{k}$.}
We note that $\bar{d}\!\geq\! \bar{k}$ is not an intrinsic condition for the rack-aware storage model. Because the repair of a failed node needs helper nodes both from inside the host rack and the entire racks outside, thus assuming $\bar{d}<\bar{k}$ will not lead to a contradiction that the data file can be retrieved from any less than $k$ nodes. Moreover, restricting to $\bar{d}<\bar{k}$ brings more benefits. Firstly, it provides more flexibility in the selection of helper racks. Secondly, it reduces the repair degree, i.e., the number of (cross-rack) connections needed for the repair, as the locally repairable codes (LRCs) do \cite{Gopalan2012}. Thirdly, it turns out that array codes with very low sub-packetization suffice to achieve the optimal repair bandwidth as $\bar{d}<\bar{k}$ (see Section \ref{sec2}-\ref{sec3}).

\vspace{-10pt}
\subsection*{(b) Reducing the intra-rack connections by introducing a threshold $l<u$.}
Instead of requiring all the remaining $u-1$ nodes in the host rack to participate in the repair, we introduce a threshold $l<u$ within the host rack such that $l$ surviving nodes in the host rack as well as $\bar{d}$ helper racks are sufficient to accomplish the node repair. That is, the initial setting of RRCs corresponds to $l=u-1$, while we extend to $0\leq l<u$.

\vspace{8pt}
After the relaxations, for all the failure patterns that node failures happen in at most $\bar{n}-\bar{d}$ racks and each rack contains no more than $u-l$ failures, our repair process is still implementable. For example, set $u\!=\!5, \bar{n}\!=\!30$ and $n\!-\!k\!=\!6$. Then $n\!=\!150$ and $k\!=\!144=5\times 28+4$, i.e., $\bar{k}=28$. Let $\bar{d}=8$ and $l=3$. If each rack contains at most two failures, our repair can tolerate $44$ node failures. Even in the worst case, our repair process still can tolerate $\min\{u-l,\bar{n}-\bar{d}\}$ node failures. Certainly, the improvement of erasure tolerance is at the sacrifice of storage. By the constructions given in this paper, we can build codes with storage overhead $\frac{n\alpha}{B}\approx 1.46$ (or $1.56$) for the example. However, the codes are very practical because they are built over $\mathbb{F}_{2^8}$ and have sub-packetization $\alpha=1$ (or $\alpha=8$). So the sacrifice of storage is acceptable and worthwhile.
For other failure patterns, we currently depend on the naive approach for repair (i.e., recover the data file from $k$ surviving nodes and then regenerate the failed nodes).

In this work, we focus on optimizing the repair bandwidth in the relaxed model. For simplicity, the code achieving the optimal repair bandwidth in this model is termed an $(\bar{n}u,k,\bar{d},l;\alpha,\beta,B)$ MET-RRC (multiple-erasure-tolerance RRC), and the MET-RRCs with the minimum storage and minimum repair bandwidth are called MET-MSRR and MET-MBRR respectively.

\subsection{Related Work}
Regenerating codes that can repair multiple node failures in rack-aware (or clustered) storage systems were also studied in \cite{Prakash17,Gupta,Wangjj}. In more detail, Abdrashitov et al. \cite{Prakash17} considered the repair of $t\geq 2$ nodes within a cluster. The main difference is that in \cite{Prakash17} (also in \cite{Prakash18}) the data file can be retrieved by any $\bar{k}$ clusters rather than any $k$ nodes.
As a result, their codes only have fault tolerance $\bar{n}-\bar{k}$, i.e., when more than $\bar{n}-\bar{k}$ nodes fail, the data file may be lost. By contrast, our codes have fault tolerance $n-k\approx (\bar{n}-\bar{k})u$.
Gupta et al. \cite{Gupta} considered the repair of multiple node failures that are evenly distributed in $f>1$ racks, so a cooperative repair approach was applied. However, they still assumed $\bar{d}\geq\bar{k}$ which causes high storage overhead because it must hold $\bar{d}\leq \bar{n}-f$. In this work we mainly consider the repair of multiple node failures in one rack. For failures distributed in different racks, we adopt a separate repair approach. The joint repair of  multiple node failures in different racks is one of our future research problems. Recently, Wang et al. \cite{Wangjj} constructed MSRR codes that can optimally repair $h\geq 1$ failed nodes within a rack from $\bar{d}$ helper racks including $\bar{e}$ corrupted helper racks. However, they still assumed $\bar{d}-2\bar{e}\geq \bar{k}$ which leads to high storage overhead and exponential sub-packetization.

Unlike our rack-aware storage model where the intra-rack communication cost is completely neglected, the models in \cite{Sohn18,Sohn19} assume a fixed ratio between the intra-rack and cross-rack communication cost and define the repair bandwidth as the sum of the two parts. Besides, they calculated the cross-rack bandwidth with respect to nodes rather than racks outside the host rack.

\subsection{Contributions}
We first derive a cut-set bound under the relaxed repair model which then induces a tradeoff between the storage per node (i.e., $\alpha$) and  repair bandwidth (i.e., $\bar{d}\beta$). The codes with parameters lying on the tradeoff curve are called MET-RRCs. Furthermore, we characterize the parameters at the two extreme points on the tradeoff curve that correspond to the MET-MSRR code and MET-MBRR code respectively.

More importantly, we build explicit MET-MSRR codes and MET-MBRR codes with the lowest sub-packetization (because $\beta=1$) for all parameters over a finite field $F$ satisfying $u\mid (|F|-1)$ and $|F|>n$. In particular, we discuss the relation between our MET-MSRR code and optimal LRCs. When $\bar{d}=0$, our construction of MET-MSRR codes turns out to be a reformulation of the construction of optimal $(r,\delta)$ LRCs in \cite{LRCFamily2014} from the perspective of parity-check matrices. As for multiple erasure tolerance, our codes permit simultaneously repairing $h\leq u-l$ node failures in one rack from $l$ local helper nodes and $\bar{d}$ helper racks with the optimal repair bandwidth. Additionally, for the convenience of practical use, we establish systematic encoding processes for the two codes.

The remaining of the paper is organized as follows. Section II derives the cut-set bound and the parameters for the two extreme points. Section III and IV present the explicit constructions of the MET-MSRR code and MET-MBRR code respectively. A systematic encoding process is provided right after each code construction. Section V concludes the paper.

\section{The cut-set bound}\label{sec1}

First we introduce some notations used throughout the paper. For nonnegative integers $i<j$, let $[j]=\{1,...,j\}$ and $[i,j]=\{i,i+1,...,j\}$. For simplicity, we label the $\bar{n}$ racks by $e\in[0,\bar{n}-1]$ and the $u$ nodes within each rack by $g\in[0,u-1]$. Thus each of the $n$ nodes is labeled by a pair $(e,g)\in[0,\bar{n}-1]\times[0,u-1]$.

Then we draw an information flow graph describing the node repair process and the file reconstruction in the rack-aware storage system. An illustration is given in Fig. \ref{fg11}.

\begin{itemize}
\item The data file consisting of $B$ symbols flows from the source vertex S to a data collector C connecting to arbitrary $k$ nodes. In order to show the storage size of each node, we split each node into two nodes  $X^{\rm in}_{\scriptscriptstyle{(e,g)}}$ and $X^{\rm out}_{\scriptscriptstyle{(e,g)}}$ with a directed edge of capacity $\alpha$.
\item When a node $X_{\scriptscriptstyle{(e,g)}}$ fails, the replacement node $X'_{\scriptscriptstyle{(e,g)}}$ connects to $l$ surviving nodes in rack $e$ with edges of capacity $\infty$ and $\bar{d}$ helper racks outside rack $e$ with edges of capacity $\beta$. After the repair, the whole rack is replaced by a replacement rack and the unrepaired nodes are copied to the replacement rack with edges of capacity $\infty$ from the original node to the copy.
\item Within each helper rack, there exist directed edges of capacity $\infty$ from the other nodes to the node connected by the replacement node, which means the $\beta$ symbols uploaded by the helper rack are computed based on the data of all nodes in the rack.

\end{itemize}

\begin{figure}[ht]
\begin{center}
\includegraphics[width=0.65\columnwidth]{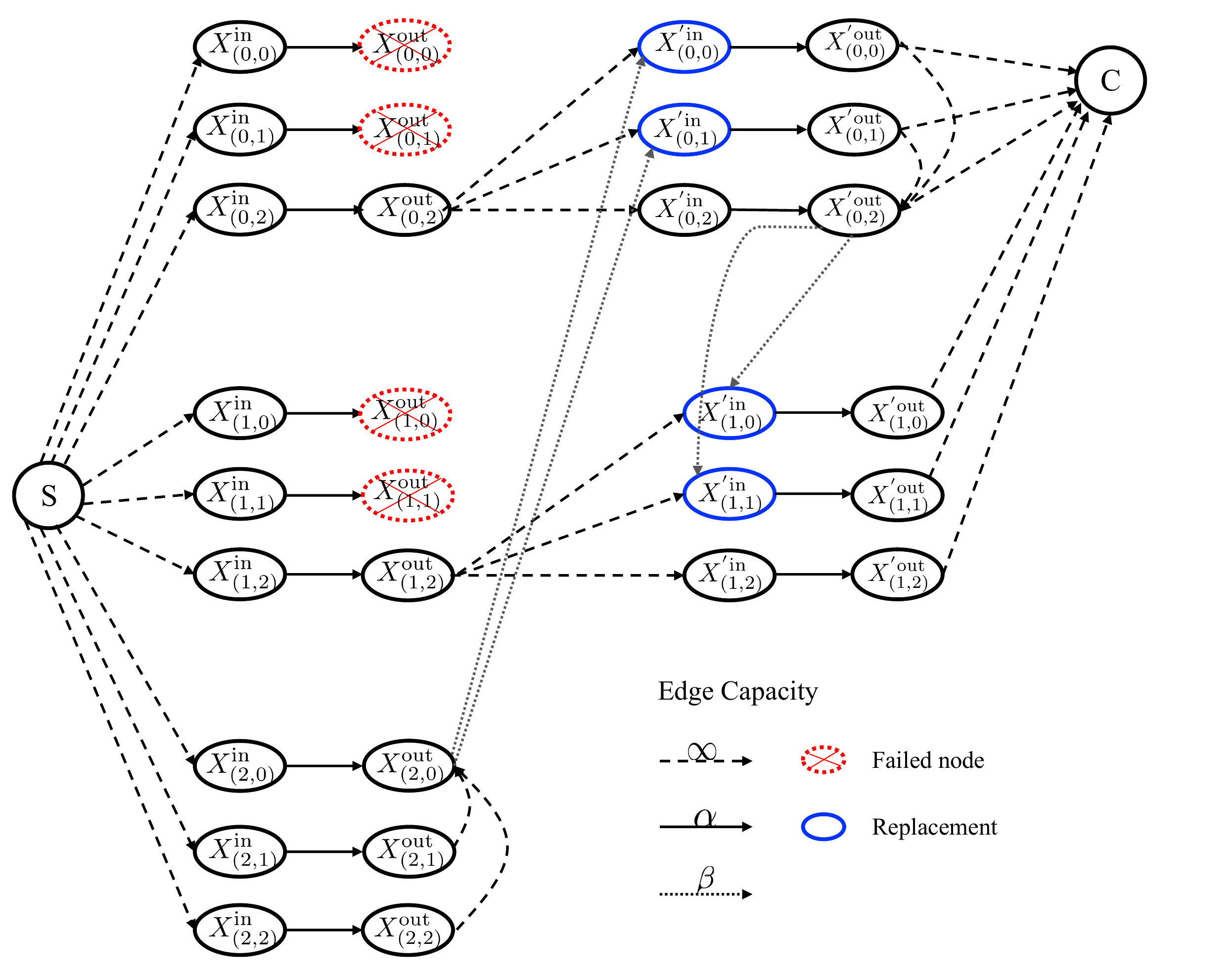}
\end{center}
\caption{\scriptsize An example of the information flow graph with $(n,k,u,\bar{d},l)=(9,6,3,1,1)$.}\label{fg11}
\end{figure}

\begin{theorem}\label{thm1}(The cut-set bound) For an $(\bar{n}u,k,\bar{d},l;\alpha,\beta,B)$ rack-aware storage system, let $B^*$ denote the maximum possible file size, then
\begin{equation}\label{eq1}
B\leq B^*\!=\!(\bar{k}l+\min\{u_0,l\})\alpha+(u\!-\!l)\sum_{i=1}^{\bar{d}}\min\{i\beta,\alpha\}\;,	
\end{equation}
where $u_0=k-\bar{k}u$.	
\end{theorem}
\begin{proof}
By network coding for the multicast problem\cite{Yeung08}, $B^*\!=\!\min_{\scriptscriptstyle{\rm C}}$MinCut(S,C). Therefore, it is sufficient to figure out a data collector ${\rm C}^*$ that achieves the minimum MinCut(S,C) among all data collectors.

For simplicity, we divide all nodes into two classes: the replacement node (eg., $X'_{(i,j)}$ for $i,j\in\{0,1\}$ in Fig. \ref{fg11}) and the original node (eg., $X'_{(i,2)}$ for $i\in\{0,1\}$ in Fig. \ref{fg11}). It is clear that in order to cut an original node apart from  S, the cut value must be $\alpha$. However, to cut a replacement node apart from  S, there are two ways. One is to cut the edge between the 'in' node and 'out' node which has capacity $\alpha$, the other is to cut its connections with the helper nodes. Although the latter cut involves cutting off $l$ nodes from inside the rack and $\bar{d}$ racks from outside which seemingly leads to a cut value much larger than $\alpha$, the real cut value is much smaller when some of the helper nodes have already been cut off. To be specific, we have the following three claims which characterize the ${\rm C}^*$ by steps.

{\it Claim 1. Suppose ${\rm C}^*$ connects to $l'\leq l$ nodes from one rack. Then cutting off each of the $l'$ nodes contributes $\alpha$ to the mincut.}

Proof of Claim 1. For $1\leq i\leq l'$, if the $i$-th node connected by ${\rm C}^*$ in one rack is an original node, then it obviously contributes $\alpha$ to the mincut. Suppose the $i$-th node is a replacement node and cut it off from its helper nodes. The smallest possible cut is treating the former $i-1$ nodes as intra-rack helper nodes and then cutting off $l-i+1$ more intra-rack helper nodes which contributes at least $\alpha$ to the mincut. Since the cut between the 'in' node and 'out' node always contributes $\alpha$, the smallest cut for cutting off each node is $\alpha$.

Based on Claim 1 we conclude that if ${\rm C}^*$ connects to $l'\leq l$ nodes from one rack, these nodes contribute $l'\alpha$ to the mincut. Next we need to analyse the case that ${\rm C}^*$ connects to more than $l$ nodes from one rack.

{\it Claim 2. Suppose ${\rm C}^*$ connects to $i$ racks and $l'\geq l$ nodes from the $(i+1)$-th rack. Then for $1\leq j\leq l$, the $j$-th node from the $(i+1)$-th rack contributes $\alpha$ to the mincut. For $l<j\leq l'$ the $j$-th node from the $(i+1)$-th rack contributes $\min\{(\bar{d}-i)^{+}\beta,\alpha\}$ to the mincut, where $(\bar{d}-i)^{+}\triangleq\max\{\bar{d}-i,0\}$.}

Proof of Claim 2. For $1\leq j\leq l$, the claim follows from Claim 1. For $l<j\leq l'$, the $j$-th node from the $(i+1)$-th rack contributes less than $\alpha$ to the mincut only if it is a replacement node and cut off from its helper nodes. Because the first $i$ racks and the first $l$ nodes from the $(i+1)$-th rack can server as helper nodes, it only needs to cut off $\bar{d}-i$ additional helper racks for $i\leq \bar{d}$. For $i>\bar{d}$, all helper nodes have already been cut off, thus the $j$-th node contributes zero to the mincut. The claim is proved.

From Claim 2, we conclude that after connecting to $i$ racks, the rest $k-iu$ nodes concentrating on racks are no worse than separating on racks to achieve MinCut(S,${\rm C}^*$). In more detail, since $\min\{(\bar{d}-i)^{+}\beta,\alpha\}\leq \alpha$, whenever ${\rm C}^*$ has more than $u$ nodes to connect, assembling $u$ of these nodes in one rack will give a mincut no larger than other cases.
Therefore, we have the following claim.

{\it Claim 3. Suppose ${\rm C}^*$ has connected to $i$ racks. If $k-ui\geq u$, we can assume ${\rm C}^*$ connects to one more rack without changing MinCut(S,${\rm C}^*$).}

As a result, we can always assume ${\rm C}^*$ connects to $\bar{k}$ racks and $u_0$ nodes from the $(\bar{k}+1)$-th rack.

Next we compute MinCut(S,${\rm C}^*$). By Claim 1 and Claim 2 we know for $1\leq i\leq\bar{d}$, the $i$-th rack contributes $l\alpha+(u-l)\min\{(\bar{d}-i+1)\beta,\alpha\}$. For $\bar{d}<i\leq \bar{k}$, the $i$-th rack contributes $l\alpha$. So the $\bar{k}$ racks in all contribute
$$\bar{k}l\alpha+(u-l)\sum_{i=1}^{\bar{d}}\min\{i\beta,\alpha\}\;.$$
Then consider the final $u_0$ nodes in the $(\bar{k}+1)$-th rack. If $u_0\leq l$, by Claim 1 the $u_0$ nodes together contribute $u_0\alpha$ to the mincut. If $u_0>l$, by Claim 2 and the condition $\bar{d}<\bar{k}$, the $u_0$ nodes together contribute $l\alpha$. So the theorem is proved.
\end{proof}

\begin{remark}
Removing the restriction $\bar{d}<\bar{k}$, the proof of Theorem \ref{thm1} actually induces a general cut-set bound
\begin{equation}\label{eq0}
B\leq (\bar{k}l+\min\{u_0,l\})\alpha+(u\!-\!l)\sum_{i=1}^{\min\{\bar{d},\bar{k}\}}\min\{(\bar{d}-i+1)\beta,\alpha\}\;.	
\end{equation}
When $\bar{d}<\bar{k}$, it is exactly the bound (\ref{eq1}) discussed in this work. When $\bar{d}\geq\bar{k}$ and $l=u-1$, (\ref{eq0}) coincides with the cut-set bound derived in \cite{Hou19}. Furthermore, when $u=1$ the bound (\ref{eq0}) degenerates into the initial cut-set bound derived in \cite{Dimakis2011}.
\end{remark}

The code with parameters meeting (\ref{eq1}) with equality are termed MET-RRC. Denote $\tilde{u}_0=\min\{u_0,l\}$.
Since $\min\{i\beta,\alpha\}\leq \alpha$, from (\ref{eq1}) it immediately has $B\leq (\bar{k}l+\tilde{u}_0)\alpha+(u-l)\bar{d}\alpha$. Thus we derive a lower bound on $\alpha$, i.e.,
\begin{equation}\label{eq2}
 \alpha\geq B/(\bar{k}l+\tilde{u}_0+(u-l)\bar{d})\;.	
\end{equation}
On the other hand, since $\bar{d}$ helper racks each providing $\beta$ symbols along with $l$ helper nodes in the host rack each providing $\alpha$ symbols can recover the failed node (plus the $l$ helper nodes within the host rack), it follows $\bar{d}\beta+l
\alpha\geq (l+1)\alpha$, namely, $\alpha\leq\bar{d}\beta$.
Combining with (\ref{eq1}) and the fact $\min\{i\beta,\alpha\}\leq i\beta$, it has
$B\leq (\bar{k}l+\tilde{u}_0)\bar{d}\beta+(u-l)\beta\sum_{i=1}^{\bar{d}}i$ which induces a lower bound on $\beta$, i.e.,
\begin{equation}\label{eq3}
\beta\geq B\big/\bar{d}(\bar{k}l+\tilde{u}_0+(u-l)(\bar{d}+1)/{2})\;.
\end{equation}
However, given (\ref{eq1}) holding with equality, (\ref{eq2}) and (\ref{eq3}) cannot hold with equality simultaneously, because the equality in (\ref{eq2}) holds only if $\alpha\leq\beta$ while equality in (\ref{eq3}) holds only if $\alpha\geq\bar{d}\beta$. Therefore, there exists a tradeoff between $\alpha$ and $\beta$ for MET-RRCs. If we optimize $\alpha$ first with respect to (\ref{eq2}) and then $\beta$ with respect to (\ref{eq1}), then we get the extreme point on the tradeoff curve corresponding to the MET-MSRR codes. Reversing the order of optimization, one can get another extreme point corresponding to the MET-MBRR codes.

\begin{corollary}\label{coro2}
The MET-MSRR and MET-MBRR codes have the following parameters: for $0<\bar{d}<\bar{k}$
\begin{equation}\label{msrr}
\alpha_{\scriptscriptstyle{\textsf{MET-MSRR}}}=\beta_{\scriptscriptstyle{\textsf{MET-MSRR}}}=B/(\bar{k}l+\tilde{u}_0+(u-l)\bar{d})\;,
\end{equation}
\begin{equation}\label{mbrr}
\alpha_{\scriptscriptstyle{\textsf{MET-MBRR}}}=\bar{d}\beta_{\scriptscriptstyle{\textsf{MET-MBRR}}}=B/(\bar{k}l+\tilde{u}_0+\frac{(u-l)(\bar{d}+1)}{2})\;.
\end{equation}
Particularly when $\bar{d}=0$, it has $\beta=0$ and we only consider MET-MSRR codes which have $\alpha_{\scriptscriptstyle{\textsf{MET-MSRR}}}=B/(\bar{k}l+\tilde{u}_0)$.
\end{corollary}

\begin{remark}
Note $\bar{k}l+\tilde{u}_0+(u-l)\bar{d}=\bar{k}u+\tilde{u}_0-(u-l)(\bar{k}-\bar{d})<k$. Thus $\alpha_{\scriptscriptstyle{\textsf{MET-MSRR}}}>B/k$, which means our MET-MSRR code is not an $[n,k]$ MDS codes. Actually, because we restrict to $\bar{d}<\bar{k}$, part of the $k$ nodes can be recovered from the remains, so redundancy is already introduced among $k$ nodes. As stated before, by restricting to $\bar{d}<\bar{k}$ we sacrifice a little storage for lower repair degree and more flexibility in the repair.
\end{remark}

\section{Explicit construction of MET-MSRR codes}\label{sec2}
In this section we give an explicit construction of the MET-MSRR code for all $0\leq \bar{d}<\bar{k}$. Particularly, our code is scalar, i.e., $\alpha=\beta=1$ for $0<\bar{d}<\bar{k}$ and $\alpha=1$ for $\bar{d}=0$. Then by Corollary \ref{coro2}, it has
$$B=\bar{k}l+\tilde{u}_0+(u-l)\bar{d}=\bar{k}u+\tilde{u}_0-(u-l)(\bar{k}-\bar{d})
\;.$$
Our MET-MSRR code $\mathcal{C}_{\scriptscriptstyle{\textsf{MET-MSRR}}}$ is constructed as a $B$-dimensional subcode of an $[n,\bar{k}u+\tilde{u}_0]$ generalized Reed-Solomon (GRS) code. The details are given below.

\begin{construction}
Let $F$ be a finite field satisfying $|F|>n$ and $u\mid (|F|-1)$.
\begin{enumerate}
  \item Construct an $[n,\bar{k}u+\tilde{u}_0]$ GRS code over $F$.
  \begin{itemize}
    \item Choose $n$ distinct elements in $F$ as the locators. Let $\xi$ be a primitive element of $F$ and $\eta=\xi^{\frac{|F|-1}{u}}$. Then for $e\in[0,\bar{n}-1]$ and $g\in[0,u-1]$, define
\begin{equation}\label{eq6}
\lambda_{(e,g)}=\xi^e\eta^g\;.
\end{equation}
It can be seen $\lambda_{(e,g)}\!\neq\!\lambda_{(e',g')}$ for $(e,g)\!\neq\! (e',g')$, because $(\xi^{e-e'})^u\!\neq\! 1$ for $e\!\neq\! e'\in[0,\bar{n}-1]$ while $(\eta^{g'\!-g})^u\!=\!1$ for all $g,g'\in[0,u-1]$.
\item The GRS code is defined by  a parity-check matrix $H\in F^{(n-\bar{k}u-\tilde{u}_0)\times n}$ with the $(t, (e,g))$-entry being
\begin{equation*}
H(t,(e,g))=\lambda^t_{(e,g)},~~
t\!\in\![0,n\!-\!\bar{k}u\!-\!\tilde{u}_0\!-\!1], \; (e,g)\!\in\![0,\bar{n}\!-\!1]\!\times\![0,u\!-\!1]\;.
\end{equation*}
  \end{itemize}
\item Add more rows to $H$ to define a $B$-dimensional subcode which is the $\mathcal{C}_{\scriptscriptstyle{\textsf{MET-MSRR}}}$.
\begin{itemize}
  \item  For $i\in[0,u-l-1]$, define $T_i=\{i+ju\mid j\!\in\![\bar{n}\!-\!\bar{k},\bar{n}\!-\!\bar{d}\!-\!1]\}$. It is easy to see that $T_i\cap [0,n-\bar{k}u-\tilde{u}_0-1]=\emptyset$ and $T_i\cap T_j=\emptyset$ for $i\neq j\in[0,u-l-1]$.
  \item Denote
\begin{equation}\label{eqT}
  T=[0,n-\bar{k}u-\tilde{u}_0-1]\cup\big(\cup_{i=0}^{u-l-1}T_i\big)\;.
\end{equation}
Obviously, $|T|=n\!-\!\bar{k}u\!-\!\tilde{u}_0\!+\!(u\!-\!l)(\bar{k}\!-\!\bar{d})=n\!-\!B$ and $T\subset [0,n-1]$.
\item Extend the row index of $H$ from $[0,n\!-\!\bar{k}u\!-\!\tilde{u}_0\!-\!1]$ to $T$ and obtain a $(n-B)\times n$ matrix of full-rank which is the parity-check matrix of  $\mathcal{C}_{\scriptscriptstyle{\textsf{MET-MSRR}}}$, i.e.,
    \begin{equation}\label{PCE}
\mathcal{C}_{\scriptscriptstyle{\textsf{MET-MSRR}}}=\{{\bm c}\in F^n\mid \sum_{e=0}^{\bar{n}-1}\sum_{g=0}^{u-1}\lambda^t_{(e,g)}c_{(e,g)}=0,~ \forall t\!\in\! T\}\;.
\end{equation}
\end{itemize}
\end{enumerate}
\end{construction}

Since $\mathcal{C}_{\scriptscriptstyle{\textsf{MET-MSRR}}}$ is a subcode of an $[n,\bar{k}u+\tilde{u}_0]$ GRS code, any $k\geq \bar{k}u+\tilde{u}_0$ node obviously can recover the original file. Before proving the node repair property of $\mathcal{C}_{\scriptscriptstyle{\textsf{MET-MSRR}}}$, we give an example to illustrate our construction.

\begin{example}\label{eg1}
Suppose $n\!=\!30, u\!=\!5, k\!=\!24, l\!=\!3, \bar{d}\!=\!2$, then $\bar{n}\!=\!6,\bar{k}\!=\!4,\tilde{u}_0\!=\!3$. The $\mathcal{C}_{\scriptscriptstyle{\textsf{MET-MSRR}}}$ has $\alpha\!=\!\beta\!=\!1$,
$B=\bar{k}l+\tilde{u}_0+(u-l)\bar{d}=19$. The code is defined by a parity-check matrix $H=(\lambda_{(e,g)}^t)_{t,(e,g)}$ which is a $11\times 30$ matrix and the $(t,(e,g))$-th entry is $\lambda_{(e,g)}^t$. Specifically, the column index $(e,g)\in[0,5]\times[0,4]$ corresponds to the storage nodes, while the row index is
\begin{equation*}
  t\in T= \{0,1,2,3,4,5,6\}\cup\{10,11\}\cup\{15,16\}\;.
\end{equation*}
Obviously, all rows of $H$ are drawn from the $30\times 30$ Vandermonde matrix $(\lambda_{(e,g)}^t)_{t\in[0,29],(e,g)}$, so $H$ is of full-rank and defines a linear code of dimension $B=19$.
Moreover, $H$ can be viewed as adding $4$ rows at the bottom of a $7\times 30$ Vandermonde matrix $(\lambda_{(e,g)}^t)_{t\in[0,6],(e,g)}$, so $\mathcal{C}_{\scriptscriptstyle{\textsf{MET-MSRR}}}$ is a subcode of a $[30,23]$ GRS code. Thus $k=24$ nodes suffice to retrieve the data file.
\end{example}

\subsection{Repair of multiple nodes in one rack}
 According to the rack-aware storage model, any single node failure can be repaired from $\bar{d}$ helper racks and $l$ survival nodes in the host rack. Therefore, up to $u-l$ node failures in one rack are repairable as long as $\bar{d}$ helper racks are available outside the rack. In the section, we show that in $\mathcal{C}_{\scriptscriptstyle{\textsf{MET-MSRR}}}$ any $h\leq u-l$ node failures in one rack can be repaired in a centralized manner from $\bar{d}$ helper racks each transferring $h\beta=h$ symbols as well as $l$ survival nodes within the host rack.

\begin{definition}[Rack-level Code]\label{defRLCode}
 For each codeword $\bm{c}=(c_{(e,g)})_{e\in[0,\bar{n}-1],g\in[0,u-1]}$ in $\mathcal{C}_{\scriptscriptstyle{\textsf{MET-MSRR}}}$ and for $i\in[0,u-l-1]$, define ${\bm w}({\bm c})^{(i)}=(w^{(i)}_0,w^{(i)}_1,...,w^{(i)}_{\bar{n}-1})\in F^{\bar{n}}$ where $w^{(i)}_e=\sum_{g=0}^{u-1}\lambda_{(e,g)}^ic_{(e,g)}$ for $e\in[0,\bar{n}-1]$. Then $\mathcal{W}^{(i)}=\{{\bm w}({\bm c})^{(i)}\mid {\bm c}\in\mathcal{C}_{\scriptscriptstyle{\textsf{MET-MSRR}}}\}$ is called a rack-level code from $\mathcal{C}_{\scriptscriptstyle{\textsf{MET-MSRR}}}$.
\end{definition}

In short, the rack-level code is formed by combining every $u$ coordinates of a rack into one coordinate. In the following, one can see the rack-level code plays as a bridge in the repair of node failures.

\begin{proposition}\label{prop4}
For $i\!\in\![0,u\!-\!l\!-\!1]$, all codewords in the rack-level code $\mathcal{W}^{(i)}$ from $\mathcal{C}_{\scriptscriptstyle{\textsf{MET-MSRR}}}$ fall in the same $[\bar{n},\bar{d}]$ MDS code.
\end{proposition}
\begin{proof}
We derive a system of parity-check equations (PCEs) for $\mathcal{W}^{(i)}$ by
restricting the PCEs in (\ref{PCE}) to those with $t\in\{i+ju\mid j\in[0,\bar{n}-\bar{d}-1]\}$. Specifically, since $n-\bar{k}u-\tilde{u}_0-1=(\bar{n}-\bar{k})u-(\tilde{u}_0+1)$ and $\tilde{u}_0\in[0,l]$, it has $$\{i+ju\mid j\in[0,\bar{n}-\bar{k}-1]\}\subset [0,n-\bar{k}u-\tilde{u}_0-1]\;.$$
Combining with $T_i\subset T$, it is reasonable to restrict the PCEs to those with $t\in\{i+ju\mid j\in[0,\bar{n}-\bar{d}-1]\}$ and one can get
\begin{equation}\label{R1}
\sum_{e=0}^{\bar{n}-1}\sum_{g=0}^{u-1}\lambda_{(e,g)}^{i+ju} c_{(e,g)}=0,~j\!\in\![0,\bar{n}\!-\bar{d}\!-\!1]\;.
\end{equation}
By Definition \ref{defRLCode} and $\lambda_{(e,g)}^{ju}=(\xi^e\eta^g)^{ju}=\xi^{euj}$, (\ref{R1}) becomes \begin{equation}\label{R2}\sum_{e=0}^{\bar{n}-1}(\xi^{eu})^jw^{(i)}_e=0,~~
j\!\in\![0,\bar{n}\!-\bar{d}\!-\!1]\;.\end{equation}
Since $\xi$ is a primitive element of $F$ and $|F|>n$, $1,\xi^{u},\xi^{2u},...,\xi^{(\bar{n}-1)u}$ are distinct elements in $F$. Therefore, the PCEs in (\ref{R2}) define the same $[\bar{n},\bar{d}]$ MDS code for all $i\in[0,u-l-1]$ and the proposition follows.
\end{proof}

\begin{theorem}\label{thm5}
For any $e^*\in[0,\bar{n}-1]$, let $e^*$ be the host rack that contains $u-l$ failed nodes $\mathcal{F}=\{(e^*,g_1),...,(e^*,g_{u-l})\}$. Then
for each codeword ${\bm c}\in\mathcal{C}_{\scriptscriptstyle{\textsf{MET-MSRR}}}$, and for any $\bar{d}$ helper racks $h_1,...,h_{\bar{d}}\in [0,\bar{n}-1]\setminus\{e^*\}$, there exist vectors ${\bm s}_i\in F^{u-l}$ computed from  ${\bm c}$ restricted to the coordinates labelled by $(h_i,g)$  for all $g\in[0,u-1]$ and $i\in[\bar{d}]$, such that ${\bm c}$ restricted to the coordinates in $\mathcal{F}$, i.e.,  ${\bm c}_{\mathcal{F}}$, can be computed from ${\bm s}_1,...,{\bm s}_{\bar{d}}$ and the remaining $l$ coordinates of ${\bm c}$ in rack $e^*$.
\end{theorem}

\begin{proof}
By Definition \ref{defRLCode}, one can see that for $e\in[0,\bar{n}-1]$,
\begin{equation}\label{eq10}
  \begin{pmatrix}
    1&1&\cdots&1\\\lambda_{(e,0)}&\lambda_{(e,1)}&\cdots&\lambda_{(e,u-1)}\\
    \vdots&\vdots&\vdots&\vdots\\\lambda_{(e,0)}^{u-l-1}&\lambda_{(e,1)}^{u-l-1}&\cdots&\lambda_{(e,u-1)}^{u-l-1}
  \end{pmatrix}
  \begin{pmatrix}
    c_{(e,0)}\\c_{(e,1)}\\\vdots\\c_{(e,u-1)}
  \end{pmatrix}=\begin{pmatrix}
    w_e^{(0)}\\w_e^{(1)}\\\vdots\\w_e^{(u-l-1)}
  \end{pmatrix}\;.
\end{equation}
It is clear any $u-l$ columns of the matrix $(\lambda_{(e,g)}^i)_{i,g}$ on the left of (\ref{eq10}) form an invertible sub-matrix. Particularly when $e=e^*$, it implies that ${\bm c}_{\mathcal{F}}$ can be computed from the remaining $l$ coordinates in rack $e^*$ along with the $w_{e^*}^{(i)}$ 's for $i\in[0,u-l-1]$. Moreover, by Proposition \ref{prop4} we have that for $i\in[0,u-l-1]$, $w_{e^*}^{(i)}$ can be determined by $w_{h_1}^{(i)},...,w_{h_{\bar{d}}}^{(i)}$. That is, the vector ${\bm s}_j$ required in the theorem is defined as ${\bm s}_j=(w_{h_j}^{(0)},w_{h_j}^{(1)},...,w_{h_j}^{(u-l-1)})$ for $j\in[\bar{d}]$. Again by (\ref{eq10}) ${\bm s}_j$ clearly can be computed from the coordinates in rack $h_j$.
\end{proof}

Theorem \ref{thm5} states that any $u-l$ node failures in one rack can be repaired from $\bar{d}$ helper racks each transferring $(u-l)\beta$ symbols. Furthermore, in the next corollary we prove that any $h~(\leq u-l)$ node failures in one rack can be repaired from $\bar{d}$ helper racks each transferring $h\beta$ symbols. Moreover, within the host rack, only $l$ surviving nodes are needed for the repair. This fact brings more flexibility in determining when to start a repair process.

\begin{corollary}\label{coro6}
For any $h\in[u-l]$, replacing $u-l$ with $h$ in Theorem \ref{thm5}, the statement still holds.
\end{corollary}
\begin{proof}
Without loss of generality, we assume $\mathcal{F}=\{(e^*,0),...,(e^*,h-1)\}$  and the local helper nodes are $(e^*,u-l),...,(e^*,u-1)$. For simplicity, denote the matrix $(\lambda_{(e^*,g)}^i)_{i,g}$ on the left of (\ref{eq10}) when $e=e^*$ by $\Lambda^*$.
Since the first $u-l$ columns of $\Lambda^*$ are linearly independent, there exists a $h\times (u-l)$ matrix $A^*$ of rank $h$ such that $A^*\Lambda^*=(I_h\mid 0 \mid P)$ where $I_h$ is the $h\times h$ identity matrix, $P$ is a $h\times l$ matrix, and the middle all-zero matrix has $u-l-h$ columns. Denote $A^*(w_e^{(0)},...,w_e^{(u-l-1)})^\tau=(v_e^{(0)},...,v_e^{(h-1)})^\tau$ for all $e\in[0,\bar{n}-1]$. Then for $i\in[0,h-1]$, $(v_0^{(i)},v_1^{(i)},...,v_{\bar{n}-1}^{(i)})$ still falls in the $[\bar{n},\bar{d}]$ MDS code defined in (\ref{R2}) because it is a linear combination of ${\bm w}({\bm c})^{(0)},...,{\bm w}({\bm c})^{(u-l-1)}$ which are all codewords of the $[\bar{n},\bar{d}]$ MDS code. Then multiplying $A^*$ on both sides of (\ref{eq10}) for all $e\in[0,\bar{n}-1]$, the repair of nodes in $\mathcal{F}$ can be achieved by using $(v_0^{(i)},v_1^{(i)},...,v_{\bar{n}-1}^{(i)}), i\in[0,h-1]$ as the corresponding rack-level codewords.
\end{proof}

\begin{example}
We illustrate the node repair of the $\mathcal{C}_{\scriptscriptstyle{\textsf{MET-MSRR}}}$ given in Example \ref{eg1}. First the rack-level code $\mathcal{W}^{(i)}$, $i\in\{0,1\}$, has codewords ${\bm w}^{(i)}=(w^{(i)}_0,w^{(i)}_1,...,w^{(i)}_5)$ which are defined by
\begin{equation}\label{eqeg}\begin{pmatrix}
1&1&\cdots&1\\\lambda_{(e,0)}&\lambda_{(e,1)}&\cdots&\lambda_{(e,4)}
\end{pmatrix}\begin{pmatrix}
c_{(e,0)}\\\vdots\\c_{(e,4)}
\end{pmatrix}=\begin{pmatrix}
w_e^{(0)}\\w_e^{(1)}
\end{pmatrix}, \;\forall e\in[0,5]\;.\end{equation}
By restricting the rows of $H$ to those with $t\in\{0,5,10,15\}$ (or $t\in\{1,6,11,16\}$) one can see that ${\bm w}^{(0)}$ (or ${\bm w}^{(1)}$) falls in a $[6,2]$ GRS code. Suppose node $(0,0)$ fails and the local helper nodes are $(0,2),(0,3),(0,4)$. By the proof of Corollary \ref{coro6}, we first compute a $2\times 2$ invertible matrix $A$ such that $A\begin{pmatrix}
1&1&\cdots&1\\\lambda_{(0,0)}&\lambda_{(0,1)}&\cdots&\lambda_{(0,4)}
\end{pmatrix}=\begin{pmatrix}
1&0&*&\cdots&*\\0&1&*&\cdots&*
\end{pmatrix}$. Denote $\begin{pmatrix}
{\bm v}^{(0)}\\{\bm v}^{(1)}
\end{pmatrix}=A\begin{pmatrix}
{\bm w}^{(0)}\\{\bm w}^{(1)}
\end{pmatrix}$. Then ${\bm v}^{(i)}$, $i\in\{0,1\}$, falls in the same $[6,2]$ GRS code as ${\bm w}^{(0)}$ and ${\bm w}^{(1)}$ do. Multiply $A$ from left on both sides of (\ref{eqeg}) and restrict to the first row, then one can see $c_{(0,0)}$ can be recovered from $c_{(0,2)},c_{(0,3)},c_{(0,4)}$ and $v^{(0)}_0$. Moreover, $v^{(0)}_0$ can be recovered from any other two coordinates (corresponding to $\bar{d}=2$ helper racks) of ${\bm v}^{(0)}=(v^{(0)}_0,v^{(0)}_1,...,v^{(0)}_5)$, and for each helper rack $e$, the coordinate $v^{(0)}_e$ can be derived from all nodes in rack $e$.
\end{example}

\subsection{Relation with the optimal locally repairable codes}
Locally repairable codes (LRCs) are another kind of important codes for distributed storage besides regenerating codes, which aim to reduce the repair degree. Particularly, the optimal LRCs achieve the maximum code distance for given storage redundancy. In this section we show $\mathcal{C}_{\scriptscriptstyle{\textsf{MET-MSRR}}}$ degenerates into an LRC when $\bar{d}=0$. Moreover, since  $\mathcal{C}_{\scriptscriptstyle{\textsf{MET-MSRR}}}$ achieves the minimum storage, it turns out to be an optimal LRC. For more details, we first recall some basics of LRCs.

\begin{definition}(\cite{Prakash2012})
The $i$-th code symbol, $i\in[n]$, in an $[n,k]$ linear code $\mathcal{C}$ is said to have locality $(r,\delta)$ if there exists a subset $S_i\subseteq[n]$ such that
\begin{itemize}
  \item $i\in S_i$ and $|S_i|\leq r+\delta-1$;
  \item the minimum distance of $\mathcal{C}|_{S_i}$ is at least $\delta$.
\end{itemize}
\end{definition}

Let $d_{\min}$ denote the minimum distance of $\mathcal{C}$. It was proved in \cite{Prakash2012} that an $[n,k]$ linear code with all information symbols satisfying the locality $(r,\delta)$ must have the minimum distance
\begin{equation}\label{eq11}
  d_{\min}\leq n-k+1-(\lceil\frac{k}{r}\rceil-1)(\delta-1)\;.
\end{equation}
If a linear code with all code symbols satisfying the locality $(r,\delta)$ meets the bound (\ref{eq11}) with equality, it is called an optimal LRC. In exploring the relation between LRCs and MET-MSRR code, it is important to note the `$k$' in an $[n,k]$ linear code means the code dimension, while the `$k$' in an $(n=\bar{n}u,k,\bar{d},l)$ MET-RRC actually means the code distance $d_{\min}\geq n-k+1$ (because any $k$ nodes can recover the original file). Moreover, the $\mathcal{C}_{\scriptscriptstyle{\textsf{MET-MSRR}}}$ has code dimension $B$ rather than $k$.

\begin{theorem}\label{thm8}
When $\bar{d}=0$ and $1\leq l\leq u-1$, the $(n=\bar{n}u,k,\bar{d},l)$ $\mathcal{C}_{\scriptscriptstyle{\textsf{MET-MSRR}}}$ constructed in (\ref{PCE}) is an $[n,\bar{k}l+\tilde{u}_0]$ optimal LRC with locality $(r=l,\delta=u-l+1)$.
\end{theorem}
\begin{proof}
By the construction in (\ref{PCE}) and the definition of rack-level codes $\mathcal{W}^{(i)}$'s in Definition \ref{defRLCode}, the linear system (\ref{R2}) still holds for $\bar{d}=0$ and it implies $w_{e}^{(i)}=0$ for $e\in[0,\bar{n}-1]$. Then combining with (\ref{eq10}) we know when $\bar{d}=0$, $(c_{(e,0)},c_{(e,1)},...,c_{(e,u-1)})$ falls in an $[u,l]$ GRS code for all $e\in[0,\bar{n}-1]$. That is, the $\mathcal{C}_{\scriptscriptstyle{\textsf{MET-MSRR}}}$ restricted to each rack is an $[u,l]$ MDS code, so each code symbol has locality $(r=l,\delta=u-l+1)$. Next we verify $\mathcal{C}_{\scriptscriptstyle{\textsf{MET-MSRR}}}$ meets (\ref{eq11}) with equality in the following two cases.

(1)~$u\nmid k$ : Then $\tilde{u}_0=\min\{l,u_0\}\geq 1$. Since the $(n,k,\bar{d},l)$ $\mathcal{C}_{\scriptscriptstyle{\textsf{MET-MSRR}}}$ is an $[n,B=\bar{k}l+\tilde{u}_0]$ LRC with locality $(r=l,\delta=u-l+1)$, by (\ref{eq11}) it has $d_{\min}\leq n-(\bar{k}l+\tilde{u}_0)+1-(\lceil\frac{\bar{k}l+\tilde{u}_0}{l}\rceil-1)(u-l)=n-(\bar{k}u+\tilde{u}_0)+1$. On the other hand, since $\mathcal{C}_{\scriptscriptstyle{\textsf{MSRR}}}$ is a subcode of an $[n,\bar{k}u+\tilde{u}_0]$ GRS code, it implies $d_{\min}\geq n-(\bar{k}u+\tilde{u}_0)+1$. As a result, the bound (\ref{eq11}) is met with equality.

\vspace{4pt}
(2)~$u\mid k$ : It follows $k=\bar{k}u$ and $\tilde{u}_0=0$. Thus $\mathcal{C}_{\scriptscriptstyle{\textsf{MET-MSRR}}}$ is an $[n,B=\bar{k}l]$ LRC with locality $(r=l,\delta=u-l+1)$, and by (\ref{eq11}) it has $d_{\rm min}\leq n-\bar{k}l+1-(\lceil\frac{\bar{k}l}{l}\rceil-1)(u-l)=n-(\bar{k}-1)u-l+1$. On the other hand, from the definition of $T$ in (\ref{eqT}) one can see $[0,n-\bar{k}u+(u-l-1)]\subseteq T$. Thus in this case $\mathcal{C}_{\scriptscriptstyle{\textsf{MET-MSRR}}}$ is a subcode of an $[n,(\bar{k}-1)u+l]$ GRS code. As a result, $d_{\min}\geq n-(\bar{k}-1)u-l+1$ and the bound (\ref{eq11}) is met with equality.
\end{proof}

\begin{remark}
The optimal LRC with all-symbol locality $(r,\delta)$ implied from our construction of $\mathcal{C}_{\scriptscriptstyle{\textsf{MET-MSRR}}}$ in the case $\bar{d}=0$ turns out to be another perspective of the optimal LRCs given in \cite{LRCFamily2014}. Specifically, the authors in \cite{LRCFamily2014} constructed LRCs from the generator matrix with rows drawn from a Vandermonde matrix according to a good polynomial, while we construct code from the parity-check matrix by adding rows to a Vandermonde matrix. Both the good polynomial in \cite{LRCFamily2014} and the rack-level code defined here are used to realize locality within racks, while the Vandermonde matrices are used to control the minimum distance. However, we further design the rack-level code as an MDS code or an MBR code (in the next section) to attain the optimal cross-rack repair bandwidth.
\end{remark}

\subsection{A systematic version of the MET-MSRR code}\label{sec3c}
Systematic codes are desirable in practical use because the original data symbols are stored in some nodes in uncoded form, thus one can directly obtain the original data symbols by accessing the corresponding nodes. Particularly, for $[n,k]$ systematic MDS codes with sub-packetization $\alpha$, there exist $k$ nodes that store exactly the original $k\alpha$ data symbols. However, in MET-RRCs because $\bar{d}<\bar{k}$ is required, redundancy is introduced among any $k$ nodes for both MET-MSRR codes and MET-MBRR codes. Therefore, we give the following definition of systematic MET-RRCs.

\begin{definition}\label{def12}(Systematic MET-RRC)
 An $(\bar{n}u,k,\bar{d},l)$ MET-RRC with parameters $\alpha,\beta,B$ is called systematic if there exist $\hat{k}~(\leq k)$ nodes such that the $\hat{k}\alpha$ symbols stored on them contain the original $B$ data symbols. Moreover, the $\hat{k}$ nodes are called systematic nodes and the set of the $B$ coordinates that store exactly the original data symbols is called an information set.
\end{definition}

In the following, we first describe an information set for the $\mathcal{C}_{\scriptscriptstyle{\textsf{MET-MSRR}}}$ constructed in (\ref{PCE}) and then develop a systematic encoding process accordingly.

\begin{theorem}\label{thm13}
Let $X$ be a set of coordinates defined as follows
\begin{equation}\label{eq23}X=X_1\cup X_2\cup X_3\end{equation} where
\begin{equation*}
  X_i=\begin{cases}
  [0,\bar{d}-1]\times [0,u-1]\;,~~~~~~~~{\rm for~}i=1\;;\\
  [\bar{d},\bar{k}-1]\times [0,l-1]\;,~~~~~~~~{\rm for~}i=2\;;\\
  \{(\bar{k},g)\mid g\in[0,\tilde{u}_0-1]\}\;,~~~~{\rm for~}i=3\;.
  \end{cases}
\end{equation*}
Particularly when $\tilde{u}_0=0$, set $X_3=\emptyset$. Then,  $X$ is an information set of $\mathcal{C}_{\scriptscriptstyle{\textsf{MET-MSRR}}}$.
\end{theorem}

\begin{proof}
It can be seen that $|X|=\bar{d}u+(\bar{k}-\bar{d})l+\tilde{u}_0=B$. Then, it suffices to prove that for any ${\bm c}\in \mathcal{C}_{\scriptscriptstyle{\textsf{MET-MSRR}}}$, ${\bm c}_X={\bm 0}$ always implies ${\bm c}={\bm 0}$. Consider the rack-level codeword $\bm{w}({\bm c})^{(i)}$ defined in Definition \ref{defRLCode}, one can first conclude that $w_e^{(i)}=0$ for $e\in[0,\bar{d}-1]$ and $i\in[0,u-l-1]$ because ${\bm c}_{X_1}={\bm 0}$. Then by Proposition \ref{prop4}, $\bm{w}({\bm c})^{(i)}$ falls in an $[\bar{n},\bar{d}]$ MDS code, so $w_e^{(i)}=0$ for $e\in[0,\bar{d}-1]$ implies $w_e^{(i)}=0$ for all $e\in[0,\bar{n}-1]$. Combining with \eqref{eq10} we know $(c_{(e,0)},...,c_{(e,u-1)})$ falls in an $[u,l]$ GRS code for all $e\in[0,\bar{n}-1]$. Note  $X_2=[\bar{d},\bar{k}-1]\times [0,l-1]$. Therefore, ${\bm c}_{X_2}={\bm 0}$ implies $c_{(e,g)}=0$ for all $(e,g)\in[\bar{d},\bar{k}-1]\times [0,u-1]$. Combining with ${\bm c}_{X_3}={\bm 0}$, it follows the first $\bar{k}u+\tilde{u}_0$ coordinates of ${\bm c}$ are all zeros. Since $\mathcal{C}_{\scriptscriptstyle{\textsf{MET-MSRR}}}$ is a subcode of an $[n,\bar{k}u+\tilde{u}_0]$ GRS code, then ${\bm c}={\bm 0}$.
\end{proof}

Furthermore, fixing ${\bm c}_X$ to be the original data symbols, then one can derive a systematic encoding process for $\mathcal{C}_{\scriptscriptstyle{\textsf{MET-MSRR}}}$. The details are given in Algorithm \ref{alg}.
Note the Line 6,7,9 of Algorithm \ref{alg} all depend on solving linear systems with Vandermonde coefficient matrices. Particularly in Line 9, since the first $\bar{k}u+\tilde{u}_0$ coordinates have been determined, restricting the defining equation (\ref{PCE}) to $t\in [0,n-\bar{k}u-\tilde{u}_0-1]$, one can immediately solve the remaining coordinates of ${\bm c}$ from the $\bar{k}u+\tilde{u}_0$ known coordinates. In the following, we give an example to illustrate the systematic encoding process.

\begin{algorithm}[ht]
\caption{The systematic encoding process of $\mathcal{C}_{\scriptscriptstyle{\textsf{MET-MSRR}}}$}\label{alg}
\begin{algorithmic}[1]
\REQUIRE Original data symbols $s_1,...,s_B\in F$.
\ENSURE A codeword ${\bm c}\in \mathcal{C}_{\scriptscriptstyle{\textsf{MET-MSRR}}}$ such that ${\bm c}_X=(s_1,...,s_B)$, where $X$ is defined in (\ref{eq23}).
\STATE Set ${\bm c}_X=(s_1,...,s_B)$;
\FOR{$e\in[0,\bar{d}-1]$}
\STATE Compute ${w}^{(i)}_{e}$ for all $i\in[0,u-l-1]$ according to Definition \ref{defRLCode};
\ENDFOR
\FOR{$e\in[\bar{d},\bar{k}-1]$}
\STATE Compute ${w}^{(i)}_{e}$ for all $i\in[0,u-l-1]$ from the equation \eqref{R2};
\STATE Compute $c_{(e,g)}$ for all $g\in[l,u-1]$ from the equation \eqref{eq10};
\ENDFOR
\STATE Determine the remaining coordinates of ${\bm c}$ by the data reconstruction process.
\end{algorithmic}
\end{algorithm}

\begin{example}
Suppose $n=16, u=4, k=13, l=2, \bar{d}=2$. Then $\bar{n}=4,\bar{k}=3,\tilde{u}_0=1$, and the $\mathcal{C}_{\scriptscriptstyle{\textsf{MET-MSRR}}}$ has $\alpha=\beta=1$,
$B=\bar{k}l+\tilde{u}_0+(u-l)\bar{d}=11$. Denote the $B$ data symbols as $s_1,...,s_{11}$.
We need to determine a codeword ${\bm c}\in \mathcal{C}_{\scriptscriptstyle{\textsf{MET-MSRR}}}$ such that ${\bm c}$ restricted to the first $\bar{k}u+\tilde{u}_0=13$ coordinates, denoted as ${\bm c}_{[13]}$, has the following form
\begin{equation*}
{\bm c}_{[13]}\!=(s_1,s_2,s_3,s_4, s_5, s_6, s_7, s_8, s_9, s_{10}, *,*,s_{11}
)\;,	
\end{equation*}
where $*$ denotes the unknown coordinate to be determined later.

First, according to Definition \eqref{defRLCode}, we compute
$w_{0}^{(i)}$ and $w_{1}^{(i)}$ respectively from the symbols $s_1,s_2,s_3,s_4$ and $s_5,s_6,s_7,s_8$, i.e., $w_{0}^{(i)}=\sum_{g=0}^3\lambda_{(0,g)}^is_{g+1}$ and $w_{1}^{(i)}=\sum_{g=0}^3\lambda_{(1,g)}^is_{g+5}$, where $i\in[0,1]$.

Then by Proposition \ref{prop4} and the equation (\ref{R2}), we can recover $w^{(i)}_{2},w^{(i)}_{3}$ as follows
$$\begin{pmatrix}w^{(i)}_{2}\\w^{(i)}_{3}\end{pmatrix}=\begin{pmatrix}1&1\\\ \xi^{2u}&\xi^{3u}\\
\end{pmatrix}^{-1}\begin{matrix}\begin{pmatrix}-w^{(i)}_{0}-w^{(i)}_{1}\\-w^{(i)}_{0}-\xi^{u}w^{(i)}_{1}\end{pmatrix}\end{matrix}\;,~~~\forall i\in[0,1]\;.$$

So far, we have obtained the rack-level codeword ${\bm w}({\bm c})^{(i)}=(w_0^{(i)},w_1^{(i)},w_2^{(i)},w_3^{(i)})$ associated with ${\bm c}$ for $i\in[0,1]$.
Then by the equation \eqref{eq10}, one has
\begin{equation}\label{eg2}
  \begin{pmatrix}
    1&1&1&1\\\lambda_{(e,0)}&\lambda_{(e,1)}&\lambda_{(e,2)}&\lambda_{(e,3)}
  \end{pmatrix}
  \begin{pmatrix}
    c_{(e,0)}\\c_{(e,1)}\\c_{(e,2)}\\c_{(e,3)}
  \end{pmatrix}=\begin{pmatrix}
    w_e^{(0)}\\w_e^{(1)}
  \end{pmatrix}\;,~~\forall e\in[0,3]\;.
\end{equation}
Particularly when $e=2$, one can derive $c_{(2,2)}$ and $c_{(2,3)}$ from (\ref{eg2}) because $c_{(2,0)}=s_9, c_{(2,1)}=s_{10}$ and $w_2^{0},w_2^{(1)}$ are already obtained. Note $c_{(2,0)},c_{(2,1)}$ are exactly the unknown coordinates in ${\bm c}_{[13]}$. As a result, $\bar{k}u+\tilde{u}_0=13$ coordinates of the codeword ${\bm c}$ have been determined, one can continue to determine the remaining coordinates through a data reconstruction process of a $[16,13]$ GRS code.
\end{example}

\section{Explicit construction of MET-MBRR codes} \label{sec3}
In this section we give an explicit construction of the MET-MBRR code. Our construction works for the scalar case $\beta=1$. Then by (\ref{mbrr}) it has
\begin{equation}\label{eq133}
\alpha=\bar{d},~B=\bar{d}(\bar{k}l+\tilde{u}_0+\frac{(u-l)(\bar{d}+1)}{2})\;.
\end{equation}
Note unlike the MET-MSRR code where each node $(e,g)$ just stores a symbol $c_{(e,g)}\in F$ (because $\alpha=1$), in the MET-MBRR code each node stores a vector  ${\bm c}_{(e,g)}\in F^{\bar{d}}$. Hereafter we use bold letters to denote row vectors.

\begin{construction}
Let $F$ be a finite field satisfying $|F|>n$ and $u\mid (|F|-1)$.
\begin{enumerate}
  \item Construct a $(\bar{k}u+\tilde{u}_0)\times \bar{d}$ matrix $M$ storing the $B$ data symbols.
  \begin{itemize}
  \item Divide the range $[0,\bar{k}u+\tilde{u}_0-1]$ into $u-l+1$ subsets, i.e. $[0,\bar{k}u+\tilde{u}_0-1]=\cup_{i=0}^{u-l}I_i$ where
\begin{equation}\label{eq13}
  I_i=\begin{cases}
  \{\delta u+(i+l-1)\mid 0\leq \delta\leq \bar{k}-1\}~~~~{\rm for~}i\in[u-l]\\
  [0,\bar{k}u+\tilde{u}_0-1]\setminus\bigcup_{i=1}^{u-l}I_i~~~~{\rm for~}i=0
  \end{cases}
\end{equation}
It is easy to check $I_i\subseteq[0,\bar{k}u+\tilde{u}_0-1]$ and $|I_0|=\bar{k}l+\tilde{u}_0$.
\item Label the rows of $M$ by indices from $[0,\bar{k}u+\tilde{u}_0-1]$, and denote by $M_i$ the matrix $M$ restricted to the rows with labels in $I_i$. Moreover, \begin{equation}\label{eq14}M_i=\begin{pmatrix}
S_i\\{\bm 0}
\end{pmatrix}\in F^{\bar{k}\times\bar{d}}{\rm~~for~}i\in[u-l]{\rm~~and ~~}M_0\in F^{(\bar{k}l+\tilde{u}_0)\times\bar{d}}\end{equation}
where $S_i$ is $\bar{d}\times \bar{d}$ symmetric matrix containing $\bar{d}(\bar{d}+1)/2$ data symbols, and ${\bm 0}$ is a $(\bar{k}-\bar{d})\times \bar{d}$ all-zero matrix. Therefore, $M_i$'s for $i\in[u-l]$ contain $(u-l)(\bar{d}+1)\bar{d}/2$ data symbols in total, and $M_0$ contains the remaining $\bar{d}(\bar{k}l+\tilde{u}_0)$ data symbols.
\end{itemize}
\item Define an $n\times (\bar{k}u+\tilde{u}_0)$ matrix $\Lambda$: for $(e,g)\in[0,\bar{n}\!-\!1]\!\times\![0,u\!-\!1]$ and $j\in[0,\bar{k}u+\tilde{u}_0-1]$, the $((e,g),j)$-th entry of $\Lambda$ is $\lambda_{(e,g)}^j$, where $\lambda_{(e,g)}$ is defined as in (\ref{eq6}).
\item Represent each codeword of the MET-MBRR code $\mathcal{C}_{\scriptscriptstyle{\textsf{MET-MBRR}}}$ by an $n\times \bar{d}$ matrix $C$ where each node stores a row of $C$. Then
    \begin{equation}\label{eq17}
\mathcal{C}_{\scriptscriptstyle{\textsf{MET-MBRR}}}=\{\Lambda M\in F^{n\times\bar{d}}\mid M {\rm ~ constructed~in~(\ref{eq14})~for~all~}B{\rm ~data~symbols}\}
\end{equation}
\end{enumerate}
\end{construction}

We give an example to illustrate the placement of the $B$ data symbols in the matrix $M$.
\begin{example}
Suppose $n=16,u=4,k=13,\bar{d}=2,l=2$. Then $\bar{n}=4,\bar{k}=3$ and $\tilde{u}_0=1$. According to (\ref{eq133}), it has $B=20$. Then we put $20$ data symbols, denoted as $s_1,...,s_{20}$, into a $13\times 2$ matrix $M$. By (\ref{eq13}) and (\ref{eq14}), we may assume the transpose of $M$, i.e., $M^\tau$, has the following form
\begin{equation*}
\setlength{\arraycolsep}{1pt}
M^\tau\!=\!\left(\begin{array}{cccc|cccc|cccc|c}
s_1&s_3&\boxed{s_5}&{s_7}&s_9&s_{11}&\boxed{s_6}&{s_8}&s_{15}&s_{17}&\boxed{0}&{0}&s_{19}\\
s_2&s_4&\boxed{s_6}&{s_8}&s_{10}&s_{12}&\boxed{s_{13}}&{s_{14}}&s_{16}&s_{18}&\boxed{0}&{0}&s_{20}
\end{array}\!\right)\;,
\end{equation*}
where for convenience, we represent the entries of $M_1$ by framed elements. It is easy so see
\begin{equation*}
\setlength{\arraycolsep}{1pt}
M_1^\tau=\left(\!\begin{array}{ccc}s_5&s_6&0\\s_6&s_{13}&0\end{array}\right),~~
M_2^\tau=\left(\!\begin{array}{ccc}s_7&s_8&0\\s_8&s_{14}&0\end{array}\right)\;.
\end{equation*}
\end{example}

Then we prove for each codeword  $C=\Lambda M\in\mathcal{C}_{\scriptscriptstyle{\textsf{MET-MBRR}}}$, any $k$ rows of $C$ are sufficient to recover the matrix $M$, namely, any $k$ nodes can recover the original data file. Actually, since $k\geq \bar{k}u+\tilde{u}_0$, from any set of $k$ nodes we arbitrarily choose a subset containing $\bar{k}u+\tilde{u}_0$ nodes. Denote this subset by $\mathcal{R}$ and let $C_{\mathcal{R}}$ be the submatrix of $C$ restricted to the rows in $\mathcal{R}$. Then $C_{\mathcal{R}}=\Lambda_{\mathcal{R}}M$. From the definition of $\Lambda$ one can see $\Lambda_{\mathcal{R}}$ is an invertible Vandermonde matrix. Therefore $M=\Lambda_{\mathcal{R}}^{-1}C_{\mathcal{R}}$. Next we prove the repair property of $\mathcal{C}_{\scriptscriptstyle{\textsf{MET-MBRR}}}$.

\subsection{Repair of multiple nodes in one rack}
Similarly, we define $u-l$ rack-level codes from $\mathcal{C}_{\scriptscriptstyle{\textsf{MET-MBRR}}}$ and then accomplish the repair by using the rack-level codes as a bridge. First, we state some basic facts about the finite field $F$:
\begin{itemize}
  \item Since $u\mid |F|-1$, it follows ${\rm Char}(F)\nmid u$, i.e., $u\neq 0$ in $F$.
  \item Because $\eta\in F$ has multiplicative order $u$, then for any integer $x$, \begin{equation}\label{eq18}
        \sum_{g=0}^{u-1}\eta^{gx}=\begin{cases}
          u~~~{\rm if~}u\mid x\\0~~~{\rm if~}u\nmid x
        \end{cases}\;.
      \end{equation}
\end{itemize}

\begin{definition}[Rack-level Code]\label{def9}
For each codeword $C\in F^{n\times\bar{d}}$ in $\mathcal{C}_{\scriptscriptstyle{\textsf{MET-MBRR}}}$ and for $i\in[0,u-l-1]$, define
${\bm w}(C)^{(i)}=({\bm w}_0^{(i)},{\bm w}_1^{(i)},...,
{\bm w}_{\bar{n}-1}^{(i)})\in (F^{\bar{d}})^{\bar{n}}$, where ${\bm w}_e^{(i)}=u^{-1}\xi^{-e(l+i)}\sum_{g=0}^{u-1}\eta^{(-l-i)g}{\bm c}_{(e,g)}$ for $e\in[0,\bar{n}-1]$ and ${\bm c}_{(e,g)}$ is the $\bar{d}$-dimensional vector formed by the $(e,g)$-th row of $C$. Then $\mathcal{W}^{(i)}=\{{\bm w}(C)^{(i)}\mid C\in\mathcal{C}_{\scriptscriptstyle{\textsf{MET-MBRR}}}\}$ is called a rack-level code from $\mathcal{C}_{\scriptscriptstyle{\textsf{MET-MBRR}}}$.
\end{definition}

\begin{proposition}\label{prop10}
For $i\in[0,u-l-1]$, all codewords in the rack-level codes $\mathcal{W}^{(i)}$ from $\mathcal{C}_{\scriptscriptstyle{\textsf{MET-MBRR}}}$ fall in  the same $(\bar{n},\bar{d},\bar{d})$ scalar (i.e., $\beta=1$) MBR code.
\end{proposition}

\begin{proof}
For each codeword ${\bm w}(C)^{(i)}=({\bm w}_0^{(i)},{\bm w}_1^{(i)},...,
{\bm w}_{\bar{n}-1}^{(i)})\in (F^{\bar{d}})^{\bar{n}}\in\mathcal{W}^{(i)}$, by Definition \ref{def9} it can be seen
$${\bm w}_e^{(i)}=u^{-1}\xi^{-e(l+i)}(1,\eta^{-l-i},...,\eta^{(u-1)(-l-i)})C_e=
u^{-1}\xi^{-e(l+i)}(1,\eta^{-l-i},...,\eta^{(u-1)(-l-i)})\Lambda_eM$$ where $C_e$ and $\Lambda_e$ respectively denote the matrix $C$ and $\Lambda$ restricted to rows labeled by $(e,g),0\leq g\leq u-1$. For $j=au+b\in[0,\bar{k}u+\tilde{u}_0-1]$ where $0\leq b<u$, the $j$-th column of $(1,\eta^{-l-i},\eta^{2(-l-i)},...,\eta^{(u-1)(-l-i)})\Lambda_e$ equals
\begin{eqnarray}
\sum_{g=0}^{u-1}\eta^{(-l-i)g}\lambda_{(e,g)}^j &=& \sum_{g=0}^{u-1}\eta^{(-l-i)g}(\xi^{e}\eta^g)^{au+b}\notag\\
 &=& \xi^{ej}\sum_{g=0}^{u-1}\eta^{(b-l-i)g}\label{eq19}\\
 &=&\begin{cases}
 u\xi^{ej}&{\rm ~~if~}b\equiv l+i {\rm ~mod~} u\\
 0 &{\rm ~~otherwise~}\label{eq20}
 \end{cases}
\end{eqnarray}
where (\ref{eq19}) comes from $\eta^u=1$ and (\ref{eq20}) follows from (\ref{eq18}).
Combining with the construction of $M$ in (\ref{eq13}) and (\ref{eq14}), we conclude that ${\bm w}_e^{(i)}=(1,\xi^{eu},...,\xi^{eu(\bar{d}-1)})S_{i+1}$. As a result, it has

\begin{eqnarray}
  \begin{pmatrix}
  {\bm w}_0^{(i)}\\{\bm w}_1^{(i)}\\\vdots\\{\bm w}_{\bar{n}-1}^{(i)}
  \end{pmatrix} &\!=\!& \begin{pmatrix}
    1&1&\cdots&1\\1&\xi^{u}&\cdots&\xi^{(\bar{d}-1)u}\\
    \vdots&\vdots&\vdots&\vdots\\\label{Sys}
    1&\xi^{(\bar{n}-1)u}&\cdots&\xi^{(\bar{n}-1)(\bar{d}-1)u}
  \end{pmatrix}S_{i+1}\label{eq211}\\
 &\!\triangleq\!&\Gamma S_{i+1}\notag
\end{eqnarray}

Note $S_{i+1}$ is symmetric and any $\bar{d}$ rows of $\Gamma$ are linearly independent. So from \cite{Kumar2011} one can see all codewords of $\mathcal{W}^{(i)}$ for $i\in[0,u-l-1]$ fall in the same $(\bar{n},\bar{d},\bar{d})$ scalar MBR code $\{\Gamma S\mid S {\rm ~is~a~}\bar{d}\times\bar{d} {\rm ~symmetric~matrix~in~}F\}$.
\end{proof}

\begin{theorem}\label{thm11}
For any $e^*\in[0,\bar{n}-1]$, let $e^*$ be the host rack that contains $h\leq u-l$ failed nodes $\mathcal{F}=\{(e^*,g_1),...,(e^*,g_{h})\}$. Then
for each codeword $C\in\mathcal{C}_{\scriptscriptstyle{\textsf{MET-MBRR}}}$, and for any $\bar{d}$ helper racks $h_1,...,h_{\bar{d}}\in [0,\bar{n}-1]\setminus\{e^*\}$, there exist vectors ${\bm s}_i\in F^h$ computed from $C$ restricted to the entries of the rows labelled by $(h_i,g)$ for all $g\in[0,u-1]$ and $i\in[\bar{d}]$, such that  $C_{\mathcal{F}}$ can be computed from ${\bm s}_1,...,{\bm s}_{\bar{d}}$ and any surviving $l$ rows of $C$ in rack $e^*$.
\end{theorem}

\begin{proof}
By Definition \ref{def9} it can be seen for $e\in[0,\bar{n}-1]$
\begin{equation}
\begin{pmatrix}
  {\bm w}_e^{(0)}\\{\bm w}_e^{(1)}\\\vdots\\{\bm w}_e^{(u-l-1)}\end{pmatrix} =\Delta C_e\;,\label{eq21}
\end{equation}
where
\begin{equation*}
\Delta=u^{-1}\begin{pmatrix}
    \xi^{-el}&&&\\&\xi^{-\!e(l+1)}&&\\&&\ddots&\\&&&\xi^{-\!e(u\!-\!1)}
  \end{pmatrix}
  \begin{pmatrix}
    1&\eta^{-\!l}&\cdots&\eta^{-\!(u\!-\!1)l}\\
    1&\eta^{-\!(l\!+\!1)}&\cdots&\eta^{-\!(u\!-\!1)(l\!+\!1)}\\
    \vdots&\vdots&\vdots&\vdots\\
    1&\eta^{-\!(u\!-\!1)}&\cdots&\eta^{-\!(u\!-\!1)^2}
  \end{pmatrix}\;.
\end{equation*}
Note any $u-l$ columns of $\Delta$ are linearly independent. Then similar to Theorem \ref{thm5}, any $u-l$ rows in $C_{e^*}$ can be computed from the remaining $l$ rows in $C_{e^*}$ and ${\bm w}_{e^*}^{(0)},...,{\bm w}_{e^*}^{(u-l-1)}$. By Proposition \ref{prop10}, ${\bm w}_{e^*}^{(i)}$ can be recovered from $\bar{d}$ helper racks each transferring one symbol which is computed from the data stored in all nodes in that rack. Furthermore, similar to Corollary \ref{coro6}, to repair any $h\leq u-l$ nodes in rack $e^*$ we can apply a linear transformation $A^*$ to both sides of (\ref{eq21}) for all $e\in[0,\bar{n}-1]$ such that $\Delta^*=A^*\Delta$ has zero columns on the non-helper nodes in rack $e^*$. Note the same linear transformation applies to all coordinates of the $u-l$ rack-level codewords, thus the resulting $h$ rack-level codewords are still in the $(\bar{n},\bar{d},\bar{d})$ MBR code and the repair process is accomplished as before.
\end{proof}

\subsection{A systematic version of the MET-MBRR code}\label{sec5}
As in Section \ref{sec3c}, in the following we provide a systematic version of the $\mathcal{C}_{\scriptscriptstyle{\textsf{MET-MBRR}}}$ defined in (\ref{eq17}) by first describing an information set and then developing a systematic encoding process.

Let $\hat{k}=\bar{k}u+\tilde{u}_0$ and we set the first $\hat{k}$ nodes to be systematic nodes. Note $\mathcal{C}_{\scriptscriptstyle{\textsf{MET-MBRR}}}$ has sub-packetization $\alpha=\bar{d}$, thus each coordinate is denoted by a tuple $[(e,g),a]$, where $(e,g)\in[0,\bar{n}-1]\times[0,g-1]$ and $a\in[0,\bar{d}-1]$. Actually, the information set $X$ is composed of the $\hat{k}\bar{d}$ coordinates of the first $\hat{k}$ nodes excluding some coordinates. The specific definition of $X$ is given in Theorem \ref{thm15}.
As an illustration, we denote the excluded coordinates by shadowed frames and display the definition of $X$ in Fig. \ref{fg2}. Note from (\ref{eq133}) we know $\mathcal{C}_{\scriptscriptstyle{\textsf{MET-MBRR}}}$ has $B=(\bar{k}u+\tilde{u}_0)\bar{d}-(\bar{k}\bar{d}-\frac{\bar{d}(\bar{d}+1)}{2})(u-l)$. One can check the number of excluded coordinates from the first $\hat{k}$ nodes in the definition of $X$ is exactly $(\bar{k}\bar{d}-\frac{\bar{d}(\bar{d}+1)}{2})(u-l)$.
\begin{figure}[H]
\begin{center}
\includegraphics[width=0.9\columnwidth]{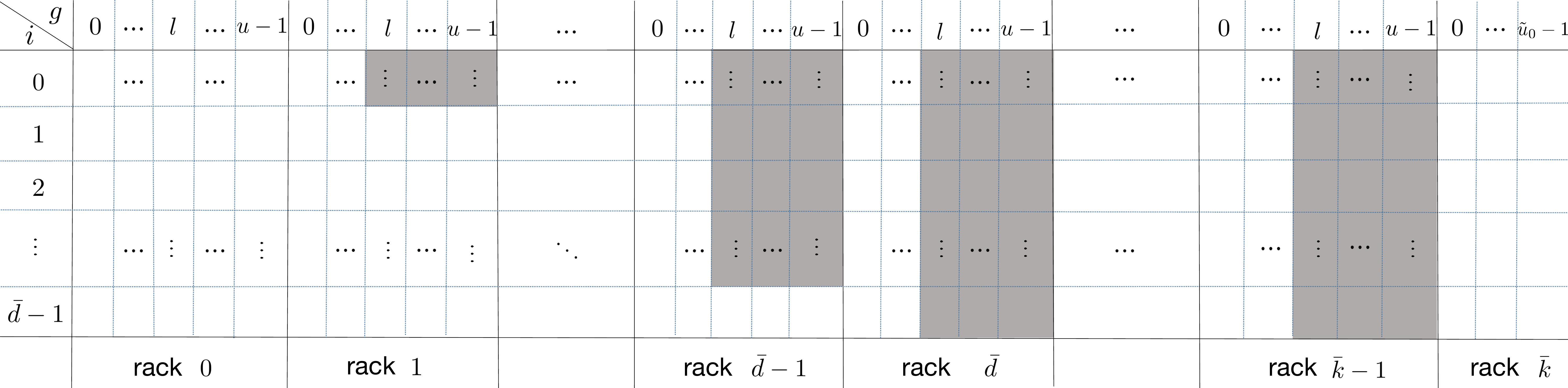}
\end{center}
\caption{The information set $X$ of $\mathcal{C}_{\scriptscriptstyle{\textsf{MET-MBRR}}}$ consists of all coordinates of the first $\hat{k}$ nodes excluding the shadowed positions.}\label{fg2}
\end{figure}

\begin{theorem}\label{thm15}
Let $X$ be a set of coordinates defined as follows
\begin{equation}\label{eq24}
X=X_1\cup X_2\cup X_3
\end{equation} where
\begin{equation*}
  X_i=\begin{cases}
  \{[(e,g),a]\mid (e,g)\in [0,\bar{k}-1]\times [0,l-1], a\in [0,\bar{d}-1]\}\;,~~~~~~~~~{\rm for~}i=1\;;\\
  \{[(e,g),a]\mid (e,g)\in [0,\bar{d}-1]\times [l,u-1], a\in [e,\bar{d}-1]\}\;,~~~~~~~~~{\rm for~}i=2\;;\\
  \{[(\bar{k},g),a]\mid g\in[0,\tilde{u}_0-1], a\in [0,\bar{d}-1]\}\;,~~~~~~~~~~~~~~~~~~~~~~~~~~{\rm for~}i=3\;.
  \end{cases}
\end{equation*}	
Then $X$ is an information set of $\mathcal{C}_{\scriptscriptstyle{\textsf{MET-MBRR}}}$.
\end{theorem}

\begin{proof}
Similar to the proof of Theorem \ref{thm13}, it suffices to prove for any codeword $C=\Lambda M\in\mathcal{C}_{\scriptscriptstyle{\textsf{MET-MBRR}}}$,  $C_X={\bm 0}$ always implies $C=0$.

For $i\!\in\![0,u\!-\!l\!-\!1]$, let ${\bm w}(C)^{(i)}=({\bm w}^{(i)}_0,...,{\bm w}^{(i)}_{\bar{n}-1})$ be the rack-level codeword associated with $C$ according to Definition \ref{def9}. Furthermore, we write ${\bm w}_{e}^{(i)}=({w}^{(i)}_{e,0},...,{w}^{(i)}_{e,\bar{d}-1})$ for $e\in[0,\bar{n}-1]$ and $i\in[0,u-l-1]$. Then from Definition \ref{def9} and the hypothesis $C_{X_1\cup X_2}=0$, we conclude ${w}^{(i)}_{e,a}=0$ for $e\in[0,\bar{d}-1]$, $i\in[0,u-l-1]$ and $a\in [e,\bar{d}-1]$. Combining with (\ref{eq211}), we have
\begin{equation}
\begin{pmatrix}1&1&\cdots&1\\1&\xi^{u}&\cdots&\xi^{(\bar{d}-1)u}\\\vdots&\vdots&\vdots&\vdots\\1&(\xi^{u})^{\bar{d}-1}&\cdots&(\xi^{(\bar{d}-1)u})^{\bar{d}-1}
\end{pmatrix}S_{i+1}=
\left(\begin{array}{cccc}
0 & 0 & \cdots & 0 \\
&0& \cdots &0\\
&& \ddots & \vdots \\
\multicolumn{2}{c}{\raisebox{1.3ex}[0pt]{\Huge*}}
&&0
\end{array}\right)\;,\label{eq25}\end{equation}
where $*$ on the right hand of (\ref{eq25}) denotes unknown entries.  Note in (\ref{eq25}) we only list the first ${\bar{d}}$ equations of the system (\ref{eq211}). Although there exist unknown entries on the right hand of (\ref{eq25}), the known entries are sufficient to recover the matrix $S_{i+1}$ because $S_{i+1}$ is a symmetric matrix. Specifically, we can recover $S_{i+1}$ column by column from right to left. For $t\in[\bar{d}]$, after recovery of the $t$ columns of $S_{i+1}$ from the right, we also obtain the lower $t$ entries of the $(t+1)$-th column (counting from the right) of $S_{i+1}$  by the symmetry. Therefore, it is enough to use the known $\bar{d}-t$ entries in the $(t+1)$-th column of the right hand of (\ref{eq25}) to recover the remaining $\bar{d}-t$ entries of the $(t+1)$-th column of $S_{i+1}$. Particularly in (\ref{eq25}), the known entries on the right hand are all zeroes, then one can finally recover $S_{i+1}=0$ for $i\in[0,u-l-1]$.

After recovery of $S_{i+1}$, we can fill up unknown positions on the right side of (\ref{eq25}). Specifically, we get ${\bm w}^{(i)}_{e}={\bm 0}$ for $e\in[0,\bar{d}-1]$ and $i\in[0,u-l-1]$. By Proposition \ref{prop10}, ${\bm w}(C)^{(i)}$ falls in an $(\bar{n},\bar{d},\bar{d})$ MBR code. As a result, ${\bm w}(C)^{(i)}={\bm 0}$ for all $i\in[0,u-l-1]$.

Then combining with \eqref{eq21}, for all $e\in[0,\bar{n}-1]$, $C_e$ falls in an $[u,l]$ MDS array code which is a direct sum of $\bar{d}$ GRS codes. From the hypothesis $C_{X_1}=0$, it implies ${\bm c}_{(e,g)}={\bm 0}$ for all $e\in[0,\bar{k}-1]$ and $g\in[0,u-1]$. Combining with the hypothesis $C_{X_3}=0$, it follows the first $\hat{k}$ nodes all store zeros. Then by the data reconstruction process of $\mathcal{C}_{\scriptscriptstyle{\textsf{MET-MBRR}}}$, one can finally recover $C=0$.
\end{proof}

Furthermore, the proof of Theorem \ref{thm15} implies a systematic encoding process for $\mathcal{C}_{\scriptscriptstyle{\textsf{MET-MBRR}}}$. The details are given in Algorithm \ref{alg2}. For any data file consisting of $B$ symbols, we first put the $B$ symbols orderly in the coordinates in $X$. Then the key step is to recover the coordinates corresponding to the shadowed positions in Fig. \ref{fg2}. For this purpose, we need to recover the symmetric matrix $S_i$ for $i\in[u-l]$ (see Line 2-6) and then compute the associated rack-level codeword (see Line 7-11).
After the key step, all coordinates of the first $\hat{k}$ nodes have been determined. Then by the data reconstruction of $\mathcal{C}_{\scriptscriptstyle{\textsf{MET-MBRR}}}$ we can recover the matrix $M$ such that $C=\Lambda M$ is the desired codeword.

\begin{algorithm}[ht]
\caption{The systematic encoding process of $\mathcal{C}_{\scriptscriptstyle{\textsf{MET-MBRR}}}$}\label{alg2}
\begin{algorithmic}[1]
\REQUIRE Original data symbols $s_1,...,s_B\in F$.
\ENSURE A codeword $C\in \mathcal{C}_{\scriptscriptstyle{\textsf{MET-MBRR}}}$ such that $C_X=(s_1,...,s_B)$, where $X$ is defined in (\ref{eq24}).
\STATE Set $C_X=(s_1,...,s_B)$;
\FOR{$e\in[0,\bar{d}-1]$}
\FOR{$a\in[e,\bar{d}-1]$}
\STATE Compute $w_{e,a}^{(i)}$ for $i\in[0,u-l-1]$ from the coordinates of $C_{X_1\cup X_2}$ according to Definition \ref{def9};
\ENDFOR
\STATE Recover the symmetric matrix $S_{i+1}$ for $i\in[0,u-l-1]$ by the system (\ref{eq211});
\STATE Compute $w_{e,a}^{(i)}$  for $i\in[0,u-l-1]$ and $a\in[0,e-1]$ from (\ref{eq211}) and $S_{i+1}$;
\ENDFOR
\FOR {$i\in[0,u-l-1]$}
\STATE Recover ${\bm w}^{(i)}_{e}$ for $e\in[\bar{d},\bar{k}-1]$ from ${\bm w}^{(i)}_{0},...,{\bm w}^{(i)}_{\bar{d}-1}$ by the node repair process of the $(\bar{n},\bar{d},\bar{d})$ MBR code defined in Proposition \ref{prop10};
\ENDFOR
\FOR{$e\in[\bar{k}-1]$}
\STATE Determine the coordinates of $C$ at the shadowed positions in Fig. \ref{fg2} by the equation (\ref{eq21});
\ENDFOR
\STATE Determine the remaining coordinates of $C$ by the data reconstruction process.
\end{algorithmic}
\end{algorithm}

\section{Conclusion}
In this work we focus on reducing the repair bandwidth of erasure codes for rack-aware storage systems in front of multiple node failures. Specifically, we propose the MET-RRC that achieves the optimal repair bandwidth when repairing $h\leq u-l$ node failures in one rack from $l$ local helper nodes and $\bar{d}$ helper nodes. Since $\bar{d}$ could be much less than the total number of racks, the multiple erasure tolerance is highly improved. Moreover, for the codes with the minimum storage and minimum bandwidth respectively, i.e., the MET-MSRR codes and MET-MBRR codes, we give explicit constructions for all parameters with the lowest sub-packetization over a field of size comparable to the code length. The MET-RRC model as well as the explicit codes provide a good option for supporting small data chunks in rack-aware (or clustered) storage systems.

\end{document}